\documentclass{patmorin}
\usepackage{graphicx}
\usepackage{amsmath,amsthm}
\usepackage{pat}
\title{\MakeUppercase{A Tight Bound on the Maximum Interference \newline
       of Random Sensors in the Highway Model}}

\author{
Evangelos Kranakis,\footnotemark[1]~\footnotemark[6]~\footnotemark[5]
\and
Danny Krizanc,\footnotemark[2]
\and
Pat Morin,\footnotemark[1]~\footnotemark[5]
\and
Lata Narayanan,\footnotemark[3]~\footnotemark[5] \\ and
\and
Ladislav Stacho\footnotemark[4]~\footnotemark[5]
}

\begin{document}

\maketitle

\footnotetext[1]{School of Computer Science, Carleton University,
Ottawa, ON, K1S 5B6, Canada.
}
\footnotetext[2]{Department of Mathematics and Computer Science,
Wesleyan University, Middletown CT 06459, USA.
}
\footnotetext[3]{Department of Computer Science and Software Engineering,
Concordia University,
Montreal, QC, H3G 1M8, Canada.
}
\footnotetext[4]{Department of Mathematics, Simon Fraser University,
Burnaby, BC,  V5A 1S6, Canada.
}
\footnotetext[5]{Research supported by NSERC (Natural Sciences and
Engineering Research Council of Canada) grant.
}
\footnotetext[6]{
Research supported by MITACS (Mathematics of Information Technology
and Complex Systems) grant.
}

\begin{abstract}
Consider $n$ sensors whose positions are 
represented by $n$ uniform, independent and identically
distributed random variables assuming values
in the open unit interval $(0,1)$. A natural way to guarantee connectivity
in the resulting sensor network is to assign to each sensor as its range, the maximum of the 
two possible distances to its two neighbors. 
The
interference at a given sensor
is defined as the number of sensors that have this sensor within their
range. In this paper we prove that 
the expected maximum interference
of the sensors is $\Theta (\sqrt{\ln n})$.
\end{abstract}

%\noindent
%{\bf Key Words and Phrases:} Interference,
%Sensors, Uniform distribution, Unit interval.

\section{Introduction}

The broadcast nature of wireless communication implies that interference with
other transmissions is inevitable. Interference can be caused by sources inside or outside the system and comes in many forms. Co-channel interference is caused by other wireless devices transmitting on the same frequency. Such 
interference can make it impossible for a receiver to decode a transmission
unless the signal power of the intended source is significantly higher than 
the combined strength of the signal received from the interfering sensors. Wireless 
devices are designed to admit a certain maximum level of interference. It is 
therefore crucial to understand the maximum possible interference that may be
experienced by any element in a 
wireless network. 

In this paper we study the expected maximum interference 
for $n$ sensors placed at random in the
{\em highway model}.
According to this model,
$n$ sensors are represented by $n$ uniform, independent and identically
distributed random variables
in the open unit interval $(0,1)$.
%Using standard notation (see \cite{order}),
%let the corresponding order
%statistics be $X_{1:n} < X_{2:n} < \cdots < X_{n:n}$. The position,  $x_i$, of
%the $i$th sensor is determined by event $X_{i:n} = x_i$ which is the
%value of the $i$th order statistic
%$X_{i:n}$. 

Since only nodes whose transmissions can reach a node can cause interference
at it, an important way to manage interference is by the use of topology control algorithms. In particular, one can assign transmission ranges to nodes with the objective of minimizing interference. On the other hand, the assignment of transmission ranges should also ensure that the network is connected. 
In the highway model, a natural algorithm is 
to assign as transmission range to a sensor the maximum distance
between its two immediate (from the left and right) neighbors since this
is the minimum range required to attain connectivity. 
%\begin{definition}
%We assign as range of the $i$th sensor,
%the maximum of the  two possible distances with its two neighbors, namely
%the random variable defined by 
%$
%\max\{ X_{i:n}-X_{i-1:n} , X_{i+1:n} -X_{i:n}  \} ,
%$
%where we use
%the notation $X_{0:n} = 0$ and  $X_{n+1:n} = 1$.
%\end{definition}
We are interested in studying the resulting {\em interference} among the
$n$ sensors. 
Intuitively, the interference for each sensor $i$ is defined 
as the number of sensors that have $i$ within their
range. 

%Let $I$ be the random variable equal to $\max_i I(i)$. We are interesting
%in obtaining bounds on $E(I)$, the expected maximum interference experienced
%by any of the sensors.

%\subsection{Related work}

Several papers study interference and network performance degradation.
Gupta and Kumar \cite{gupta2000} 
considers the throughput of wireless networks under two models of
interference: one is a protocol model that assumes interference
to be an all-or-nothing phenomenon and the other a physical model
that considers the impact of interfering transmissions on
the signal-to-noise ratio.
Motivated by this, 
Jain \etal\ \cite{jain2005} defined the concept of
{\em conflict graph} (a graph indicating which groups of nodes
interfere and hence cannot be active at the same time) and study what is the maximum 
throughput that can be supported by a wireless network given a
specific placement of wireless nodes in physical space     
and a specific traffic workload.

Burkhart \etal\ \cite{burkhart2004} proposes connectivity preserving and 
spanner constructions which are interference optimal. 
\cite{moscibroda2005} considers the {\em average interference} problem while
maintaining connectivity.
%and
%gives a sharp $O(\ln n)$ upper and lower bound.
Closely related to our study is 
the following problem first proposed in \cite{locher2008}:
\begin{quote}
Given $n$ nodes in the
plane. Connect the nodes by a spanning tree. For each node $v$ we construct a
disk centering at $v$ with radius equal to the distance to $v$'s furthest neighbor in
the spanning tree. The interference of a node $v$ is then defined as the number
of disks that include node $v$ (not counting the disk of $v$ itself). Find a spanning
tree that minimizes the maximum interference.
\end{quote}
Choosing transmission radii which minimize the
maximum interference while maintaining a connected symmetric 
communication graph 
is shown by \cite{buchin2008} to be
NP-complete.
In addition,
\cite{halldorsson2008} gives an algorithm which yields a
maximum interference in $O( \sqrt{n})$ for any set of $n$ sensors in the plane. 
%An open problem that remains is to narrow
%the gap between this upper bound and the lower bound. 
For the case of points
on a line (i.e., the highway model)
\cite{vonrickenbach2005} shows that if nodes are distributed as an exponential node chain, the algorithm described above for assigning ranges to sensors has maximum interference $\Omega(n)$. They proceed to give an $n^{1/4}$-approximation algorithm for the problem of finding an assignment
of ranges that minimizes interference.

\cite{bilo2008} shows that for 
broadcasting (one-to-all), gossiping (all-to-all), and symmetric gossiping (symmetric all-to-all)
the problem of minimizing the maximum interference experienced by
any node in the network is hard to approximate within better than
a logarithmic factor, unless NP admits slightly superpolynomial 
time algorithms. They also prove that any approximation algorithm for the problem of minimizing the total transmission power assigned to the nodes in order to guarantee any of the above communication patterns, can be transformed, by maintaining the same performance ratio, into an approximation algorithm for the problem of minimizing the total interference experienced by all the nodes in the network.

Here we study a model where sensors are represented by $n$ uniform, independent 
and identically distributed random variables in the open unit interval $(0,1)$. We 
assign to each sensor as range the maximum of the 
two possible distances with its two neighbors.
%and solve the problem
%from \cite{sirocco2010} 
For this case,
we show
a tight bound on the expected maximum interference experienced by any sensor.
In particular,
Theorem~\ref{plm}
shows that the expected maximum interference is $\Theta (\sqrt{\ln n)}$, 
with high probability. This is in contrast to the result of \cite{vonrickenbach2005}
that the maximum interference for $n$ sensors distributed on a line and
connected in the same manner is $\Omega(n)$ in the worst case.

\section{Expected Maximum Interference}
\label{section3}

%Pat's Proof
%\section{Introduction}

%We consider the following problem: A set of $n$ transmitters are
%place independently and uniformly at random on a unit interval.  Each
%transmitter, except the two extreme ones, adjust their transmission ranges
%so that they can be heard by the further of their left neighbour and right
%neighbour.  What is the maximum number of transmitters that can be heard
%by any individual transmitter? 

Let $S=\{x_1,\ldots,x_n\}$ be a set of values chosen
independently and uniformly at random from the real interval $[0,1]$
and reordered so that $x_1<\cdots<x_n$.  For each
$i\in\{2,\ldots,n-1\}$, define the \emph{broadcast range} 
\[
  R_i = \max\{x_i - x_{i-1}, x_{i+1}-x_i\}
\]
and the \emph{broadcast interval}
\[
  I_i = [x_i-R_i,x_i+R_i] \enspace . \]
For $i=1$ ($i=n$) define $R_1 = x_2 - x_1$ ($R_n = x_n - x_{n-1}$,
  respectively) and $I_1 = [ x_1 - R_1 , x_1 + R_1 ]$ ($I_n = [ x_n - R_n, x_n + R_n]$, respectively).

The \emph{interference at $x_i$} is then given by
\[
  Z_i = |\{j\in\{1,\ldots,n\}\setminus\{i\} : x_i \in I_j\}| \enspace .
\]
The \emph{maximum interference} in $S$, is given by
\[
  Z_S=\max\{Z_i:i\in\{1,\ldots,n\}\} \enspace .
\]
%We prove the following result:

%\section{The Proof}

In this section we prove our main result:
%namely that the maximum
%interference is $\Theta(\sqrt{\log n})$ with high probability.
\begin{thm}
\label{plm}
With probability $1-o(1)$, the  maximum inteference
%is in $\Theta(\sqrt{\log n})$.
$Z_S\in \Theta(\sqrt{\log n})$.
\end{thm}

This result is an immediate consequence of 
Lemmas~\ref{lower-bound}~and~\ref{upper-bound}. 
%which are proven in the next section.
Throughout this section, we will make use of the relationship between
uniformly distributed point sets and exponential random variables
\cite{d86}[Chapter~V, Theorem~2.2].  Suppose $S$ is a set of $n$ points
independently and uniformly distributed in $[0,1]$ whose elements are
$x_1,\ldots,x_n$ in sorted order.  Let $X_0,\ldots,X_n$ be Exponential(1)
random variables, let 
$x'_i=\sum_{j=0}^{i-1}X_j,$
and let $x_i''=x'_i/x'_{n+1}$.
Then $x_1'',\ldots,x_n''$ have the same distribution as $x_1,\ldots,x_n$.

Because of the above relationship we will, throughout this section, use
the convention that $X_0,\ldots,X_n$ are Exponential(1) random variables,
$x_i=\sum_{j=0}^{i-1}X_j$, and $S=\{x_1,\ldots,x_n\}$.  This definition
of $S$, $x_1,\ldots,x_n$, and $X_0,\ldots,X_n$ will be implicit in the
statements of all subsequent results and in all proofs.

\subsection{The Lower Bound}

We prove our lower-bound by defining a configuration of points that
leads to an element with interference $\Omega(\sqrt{\log n})$ and then
showing that, with high probability, this configuration occurs somewhere
in our point set.

A sequence of numbers $X_0,\ldots,X_k$ forms a \emph{$k$-frame} if
\[
     1 \le X_0 \le 2
\]
and
\[
     X_{i-1}/4 \le X_i \le X_{i-1}/2 \enspace ,
\]
for all $i\in\{1,\ldots,k\}$.  Notice that, if $X_0,\ldots,X_k$ form a
$k$-frame, then $x_{k+1}$ is a node that has interference at least $k$.
The next lemma shows that this situation is not too unlikely:

\begin{lem}\label{frame}
If $X_0,\ldots,X_k$ are a sequence of independent Exponential(1) random
variables, then the probability that $X_0,\ldots,X_k$ form a $k$-frame
is at least 
$2^{-(k+2)^2}$.
%$e^{-1}(1-e^{-1})2^{-(k+1)^2}$.
\end{lem}

\begin{proof}
Recall that an Exponential(1) random variable $X$ has cumulative
distribution function
\[
   \Pr\{X \le x\} = 1-e^{-x} \enspace .
\]
Next, observe that, in a frame,
\[
                 4^{-i} \le X_i \le 2^{-i}  \enspace ,
\]
for all $i\in\{0,\ldots,k\}$.  Let $F(X)$ be the event ``$X$ is a frame.''
Then,
\begin{eqnarray*}
     \Pr\{F(X_0,\ldots,X_{i+1}) \mid F(X_0,\ldots,X_{i})\} 
        & = & \Pr\{X_{i+1} \in [X_{i}/4,X_{i}/2] \mid F(X_0,\ldots,X_{i})\} \\
        & \ge & \Pr\{X_{i+1} \in [4^{-i}/4,4^{-i}/2]\} \\
        & = & \Pr\{X_{i+1} \in [2^{-(2i+2)},2^{-(2i+1)}]\} \\
        & = & \exp(-2^{-(2i+2)}) - \exp(2^{-(2i+1)}) \\
        & \ge & 2^{-(2i+3)} \enspace ,
\end{eqnarray*}
where the last inequality holds for all $i\ge 0$.  Therefore,
\begin{eqnarray*}
     \Pr\{F(X_0,\ldots,X_{k})\}
   & = & \Pr\{X_0\in[1,2]\}
         \cdot\prod_{i=1}^k \Pr\{F(X_0,\ldots,X_{i})
                                 \mid F(X_0,\ldots,X_{i-1})\} \\
   & = & e^{-1}(1-e^{-1})
         \cdot\prod_{i=1}^k \Pr\{F(X_0,\ldots,X_{i})
                                 \mid F(X_0,\ldots,X_{i-1})\} \\
   & = & e^{-1}(1-e^{-1})
         \cdot\prod_{i=1}^k \Pr\{X_{i} \in [X_{i-1}/4,X_{i-1}/2]\}
                                 \mid F(X_0,\ldots,X_{i-1})\} \\
   & \ge &  e^{-1}(1-e^{-1})
         \cdot\prod_{i=1}^k \Pr\{X_{i} \in [4^{-i}/4,4^{-i}/2]\} \\
   & \ge & e^{-1}(1-e^{-1})\cdot\prod_{i=1}^k 2^{-(2i+1)} \\
   & = & e^{-1}(1-e^{-1})\cdot2^{-\sum_{i=1}^k(2i+1)} \\
   & = & e^{-1}(1-e^{-1})\cdot2^{-(k^2+2k)} \\
   & \ge  & 2^{-(k+2)^2} \\
\end{eqnarray*}
as required.
\end{proof}

\begin{lem}[Lower Bound]\label{lower-bound}
With probability at least $1-\exp(-n^{1-c}/\sqrt{c\log n})$, there exists
some element of $S$ that has interference at least $\lfloor\sqrt{c\log
n}\rfloor-2$.
\end{lem}

\begin{proof}
Let $k=\lfloor \sqrt{c\log n} \rfloor-2$.  By Lemma~\ref{frame},
$X_{jk},\ldots,X_{jk+k}$ have probability at least $2^{-(k+2)^2} =
n^{-c}$ of forming a $k$-frame, in which case $x_{jk+k}$ has
interference at least $k$.  Since this is true, independently, for any
$j\in\{0,\ldots,\lfloor n/k\rfloor\}$, the probability that there is no
element of $S$ with interference greater than $k$ is at most
\[
   (1-n^{-c})^{\lfloor n/k\rfloor} \le \exp(-\lfloor n^{1-c}/k\rfloor) \enspace ,
\]
as required.
\end{proof}

\subsection{The Upper Bound}

We begin our upper-bound proof by studying a variant of interference that
is 1-sided and that considers only interference generated by transmitters
that are nearby.  The \emph{left-interference} of an element $x_t\in
S$ is the number of elements $x_i\in S$ such that $x_i < x_t$ and
$x_t-x_i \le \max\{x_i-x_{i-1},x_{i+1}-x_i\}$.  The \emph{short-range
left-interference} of $x_t$ is defined in the same way, except only
counting those elements $x_i$ such that $X_{i-1} \le 1$. (Note that this
implies $x_t-x_i \le 1$.)

\begin{lem}
\label{squeaker}
The maximum short-range left-interference of any element in $S$ is at
most $\sqrt{c\log n}$ with probability at least $1-n^{-\Omega(c)}$.
\end{lem}

\begin{proof}
We will actually prove something stronger, namely that the short-range
left-interference of any point $x \in \mathbb{R}$ is at most $\sqrt{c\log n}$
with probability at least $1-n^{-\Omega(c)}$.  We first observe that the
maximum value of the short range interference occurs when $x$ is of the
form $x_{i} + X_{i-1}$, for some $i\in\{1,\ldots,n\}$ where $X_{i-1}
\le 1$.

Consider the following process, that begins with $X_0$ and
upper-bounds the short-range left-interference at $x=x_1+X_0=2X_0$
(see Figure~\ref{fig:upper-bound}). If $X_0>1$, the process immediately ends.
Otherwise, the process proceeds in rounds where, in round $i$, there
is a length $\ell_i$.  Initially $\ell_i=X_0$.  During round~$i$, we
generate $X_{r_{i-1}+1},\ldots,X_{r_i}$ until $\sum_{j=1}^r X_{r_{i-1}+j}
\ge \ell_i/2$.  If $\sum_{j=1}^r X_j \ge \ell_i$, then the process ends.
Otherwise, we set $\ell_{i+1} = \ell_i-\sum_{j=1}^{r_i} X_j$ and continue
onto round $i+1$.

\begin{figure}
  \begin{center}
  \includegraphics{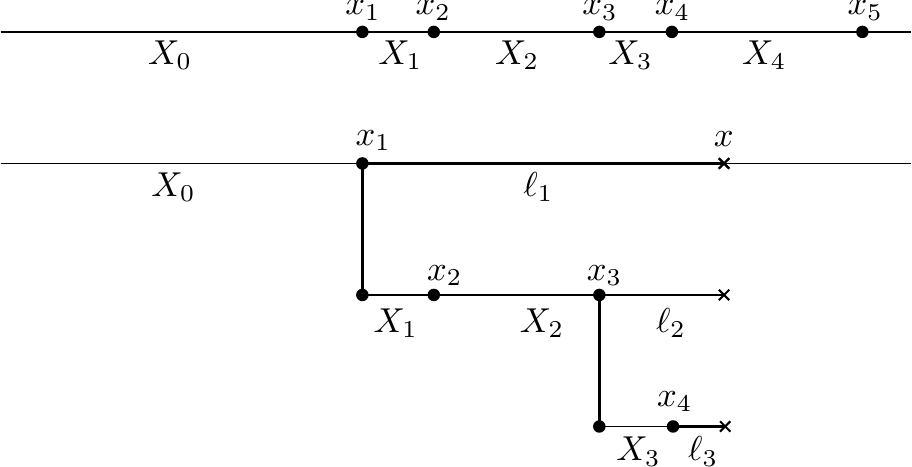}
  \end{center}
  \caption{A process that leads to an interference of 3 at $x$.  The process ends because $X_4 > \ell_3$.}
  \label{fig:upper-bound}
\end{figure}

Notice that, in this process, the only elements that might contribute
to the short-range left-interference at $x$ are $x_1$ and those $x_i$
where $X_{i-1}$ completes a round other than the final round.  Thus,
if the above process terminates during round $k$, then the short-range
left-interference at $x$ is at most $k$.

Now, observe that in round $i$, $\ell_i \le 1/2^{i-1}$.  Therefore,
the probability of continuing to round $i+1$ from round $i$ is at most
\[
   \Pr\{X_{r_{i-1}+1} \le 1/2^{i-1}\} = 1-e^{-2^{-i+1}} \le 2^{-i+1} \enspace .
\]
Therefore, the probability of continuing up to round $k$ is at most
\[
  \prod_{i=1}^{k-1} 2^{-i+1}
  = 2^{-\sum_{i=1}^{k-1}(i+1)}
  = 2^{-(k+2)(k-1)/2} \le 2^{-k^2/2} \enspace ,
\]
for $k\ge 2$.  Taking $k=\sqrt{c\log n}$, we find that this probability
is at most $1/n^{c/2}$.  Therefore, the probability that there
is \emph{any} point $x \in \mathbb{R}$ with short-range left-interference greater
than $\sqrt{c\log n}$ is at most $1/n^{c/2-1}$, as required.
\end{proof}

Finally, we have all the pieces needed to complete the upper bound:

\begin{lem}[Upper Bound]\label{upper-bound}
With probability at least $1-n^{-\Omega(c)}$, the maximum interference of
any element in $S$ is at most $\sqrt{c\log n}$.
\end{lem}

\begin{proof}
We consider only left-interference, since the right-interference
can be bounded in a symmetric way.  Consider some element $x_t$.  The
left-interference of $x_t$ is generated by some elements
$x_{i_0},\ldots,x_{i_k}$ where $x_{i_k}<\cdots<
x_{i_0}<x_t$. Lemma~\ref{squeaker} already bounds the number of elements of this
sequence where $X_{i_j-1} \le 1$.  Thus, all that remains is to bound
the number of elements $x_{i_j}$ where $X_{i_j} > 1$.

Observe, as in the proof of Lemma~\ref{squeaker}, that, for any
$j\in\{1,\ldots,k\}$, in order for $x_{i_j}$ to interfere with $x_t$
we must have
\[
   X_{i_j-1} \ge x_t - x_{i_j}
\enspace ,
\]
which implies that $X_{i_j-1} \ge 2 X_{i_{j-1}-1}$ for all
$j\in\{1,\ldots,k\}$.  Therefore, if we have $2^r$ elements with
$X_{i_j}>1$, then we have some element $X_{i_k-1} > 2^r$.  The probability
that a particular $X_i$ is greater than $2^r$ is $e^{-2^r}$. Therefore,
the probability that there exists any $X_i$ greater than $2^r$ is at most
$ne^{-2^r}$.  Setting $r=\log ( d \ln n )$ for a sufficiently large constant
$d>1$ makes this probability at most $n^{1-d}$, and completes the proof.
\end{proof}

\section{Conclusion}

In this paper we have investigated the receiver interference for
a set of random sensors on a line (also known as the highway model)
and proved a tight bound on the value of the expected maximum
interference. An interesting question would be to look
at probability distributions other than uniform for the arrangement of sensors.
%A much harder question
%is to determine the probability distribution of the receiver interference
%of $n$ random sensors on a unit interval. 
Also, bounds for the case of randomly distributed sensors in two dimensions
would be interesting. 
As in the one-dimensional case studied here where
the range of the sensors is assigned as the maximum distance to their two
neighbors, an analysis of the two-dimensional case must be preceded by
an assignment of sensor ranges. A natural choice would be assign each 
sensor a range equal the maximum distance to its neighbors in
the minimum spanning tree of the point set.

%We thank the anonymous referees for many helpful suggestions.
%This research was supported in part by Natural Sciences and
%Engineering Research Council of Canada (NSERC) and
%Mathematics of Information Technology and Complex Systems (MITACS).

\bibliographystyle{plain}
%\newpage
%\bibliographystyle{abbrvnat}
\bibliography{interference}

\begin{thebibliography}{10}

\bibitem{bilo2008}
D.~Bilo and G.~Proletti.
\newblock On the complexity of minimizing interference in ad hoc and sensor
  networks.
\newblock {\em Theoretical Computer Science}, 402:43 -- 55, July 2008.

\bibitem{buchin2008}
K.~Buchin.
\newblock Minimizing the maximum interference is hard, February, 2008.
\newblock arXiv:0802.2134.

\bibitem{burkhart2004}
M.~Burkhart, R.~Wattenhofer, and A.~Zollinger.
\newblock {Does topology control reduce interference?}
\newblock In {\em Proceedings of the 5th ACM international symposium on Mobile
  ad hoc networking and computing}, pages 9--19. ACM New York, NY, USA, 2004.

\bibitem{d86}
L.~Devroye.
\newblock {\em Non-Uniform Random Variate Generation}.
\newblock Springer-Verlag, New-York, 1986.

\bibitem{gupta2000}
P.~Gupta and P.~R. Kumar.
\newblock {The capacity of wireless networks}.
\newblock {\em Information Theory, IEEE Transactions on}, 46(2):388--404, 2000.

\bibitem{halldorsson2008}
M.M. Halld{\'o}rsson and T.~Tokuyama.
\newblock {Minimizing interference of a wireless ad-hoc network in a plane}.
\newblock {\em Theoretical Computer Science}, 402(1):29--42, 2008.

\bibitem{jain2005}
K.~Jain, J.~Padhye, V.N. Padmanabhan, and L.~Qiu.
\newblock {Impact of Interference on Multi-Hop Wireless Network Performance}.
\newblock {\em Wireless Networks}, 11(4):471--487, 2005.

\bibitem{locher2008}
T.~Locher, P.~von Rickenbach, and R.~Wattenhofer.
\newblock {Sensor Networks Continue to Puzzle: Selected Open Problems}.
\newblock {\em LNCS}, 4904:25, 2008.

\bibitem{moscibroda2005}
T.~Moscibroda and R.~Wattenhofer.
\newblock {Minimizing interference in ad hoc and sensor networks}.
\newblock In {\em Proceedings of the 2005 joint workshop on Foundations of
  mobile computing}, pages 24--33. ACM New York, NY, USA, 2005.

\bibitem{vonrickenbach2005}
P.~von Rickenbach, S.~Schmid, R.~Wattenhofer, and A.~Zollinger.
\newblock {A Robust Interference Model for Wireless Ad-Hoc Networks}.
\newblock In {\em Proc. 5th IEEE International Workshop on Algorithms for
  Wireless, Mobile, Ad-Hoc and Sensor Networks (WMAN)}, 2005.

\end{thebibliography}

\end{document}